\documentclass[11pt,letterpaper]{article}

\usepackage{standalone}
\usepackage{amsthm}
\usepackage{amsmath}
\usepackage{amssymb}
\usepackage{amsfonts}
\usepackage{algorithm}
\usepackage{mathrsfs}
\usepackage{xcolor}
\usepackage[noend]{algpseudocode}

\usepackage[
	backend=biber,
	style=alphabetic,
	sorting=ynt,
	maxbibnames=99, 
	maxcitenames=99,
	maxnames=4,
	maxalphanames=4,
	url=false,
	doi=false,
	isbn=false
]{biblatex}
\usepackage{thm-restate}

\newtheorem{theorem}{Theorem}[section]
\newtheorem{lemma}[theorem]{Lemma}

\newtheorem{definition}[theorem]{Definition}
\usepackage{thmtools}
\usepackage[hidelinks]{hyperref}
\usepackage[T1]{fontenc}

\usepackage{tabularx, environ}

\makeatletter

\newcolumntype{\expand}{}
\long\@namedef{NC@rewrite@\string\expand}{\expandafter\NC@find}

\NewEnviron{problemdefinition}[2][]{%
  \def\problemdefinition@arg{#1}%
  \def\problemdefinition@framed{framed}%
  \def\problemdefinition@lined{lined}%
  \def\problemdefinition@doublelined{doublelined}%
  \ifx\problemdefinition@arg\@empty%
    \def\problemdefinition@hline{}%
  \else%
    \ifx\problemdefinition@arg\problemdefinition@doublelined%
      \def\problemdefinition@hline{\hline\hline}%
    \else%
      \def\problemdefinition@hline{\hline}%
    \fi%
  \fi%
  \ifx\problemdefinition@arg\problemdefinition@framed%
    \def\problemdefinition@tablelayout{|>{\bfseries}lX|c}%
    \def\problemdefinition@title{\multicolumn{2}{|l|}{%
        \raisebox{-\fboxsep}{\textsc{\Large #2}}%
      }}%
  \else
    \def\problemdefinition@tablelayout{>{\bfseries}lXc}%
    \def\problemdefinition@title{\multicolumn{2}{l}{%
        \raisebox{-\fboxsep}{\textsc{\Large #2}}%
      }}%
  \fi%
  \bigskip\par\noindent%
  \begin{tabularx}{\textwidth}{\expand\problemdefinition@tablelayout}%
    \problemdefinition@hline%
    \problemdefinition@title\\[2\fboxsep]%
    \BODY\\\problemdefinition@hline%
  \end{tabularx}%
  \medskip\par%
}
\makeatother

\usepackage{tikz}
\usepackage{pgfplots}

\title{Sharp Noisy Binary Search with Monotonic Probabilities}
\author{Lucas Gretta\\UC Berkeley\footnote{Work partially done while at UT Austin} \and Eric Price\\UT Austin}

\usepackage[margin=1in]{geometry}

\newcommand{\eps}{\varepsilon}

\newcommand{\abs}[1]{{|#1|}}
\newcommand{\ceil}[1]{{\lceil #1 \rceil}}
\newcommand{\floor}[1]{{\lfloor #1 \rfloor}}

\newcommand{\E}{{\mathbb{E}}}

\newcommand{\gammaiterations}{\frac{1 + O(\gamma)}{C_{\tau, \eps}} (\lg n + O(\sqrt{\log n \log \frac{1}{\delta}} + \log \frac{1}{\delta}))}

\newcommand{\capacity}{C_{\tau,\eps}}
\newcommand{\NBSName}{BayesianScreeningSearch}
\newcommand{\SillyNBSName}{SillyBayesianScreeningSearch}

\newcommand{\bayesalg}{\textsc{BayesLearn}}
\newcommand{\mainalg}{\textsc{BayesianScreeningSearch}}
\newcommand{\mainalgsilly}{\textsc{SillyBayesianScreeningSearch}}

\newcommand{\NBS}{\textsc{MonotonicNBS}}
\newcommand{\IPNBS}{\textsc{FixedNoiseNBS}}

\DeclareMathOperator*{\argmax}{arg\,max}

\date{}

\begin{document}
\begin{titlepage}
\maketitle

\begin{abstract}
  We revisit the noisy binary search model of~\cite{karp}, in which we
  have $n$ coins with unknown probabilities $p_i$ that we can flip.
  The coins are sorted by increasing $p_i$, and we would like to find
  where the probability crosses (to within $\eps$) of a target value
  $\tau$.  This generalized the fixed-noise model of~\cite{BZ}, in
  which $p_i = \frac{1}{2} \pm \eps$, to a setting where coins near
  the target may be indistinguishable from it.  It was shown
  in~\cite{karp} that $\Theta(\frac{1}{\eps^2} \log n)$ samples are
  necessary and sufficient for this task.

  We produce a \emph{practical} algorithm by solving two theoretical
  challenges: high-probability behavior and sharp constants.  We give
  an algorithm that succeeds with probability $1-\delta$ from
  \[
    \frac{1}{C_{\tau, \eps}} \cdot \left(\lg n + O(\log^{2/3} n \log^{1/3} \frac{1}{\delta} + \log \frac{1}{\delta})\right)
  \]
  samples, where $C_{\tau, \eps}$ is the optimal such constant
  achievable.  For $\delta > n^{-o(1)}$ this is within $1 + o(1)$ of
  optimal, and for $\delta \ll 1$ it is the first bound within
  constant factors of optimal.
\end{abstract}
\thispagestyle{empty}
\end{titlepage}

\section{Introduction}

Binary search is one of the most fundamental algorithms in computer
science, finding an index $i^* \in [n]$ from $\log_2 n$ queries asking
if a given index $i$ is larger than $i^*$.  But what if the queries
are noisy?

One model for noisy binary search has each query be incorrect
independently with exactly the same probability $\frac{1}{2} - \eps$.
In this model, which we call \IPNBS{}, a line of
work~\cite{BZ,bayesoptimal,DLU21,GX} has found a sharp
bound for the required expected sample complexity,
with tight constants.
However, in many applications of noisy binary search the error
probability is not fixed, but varies with $i$: comparing $i$ to
$i^*$ is much harder when $i$ is close to $i^*$.

As one example, consider the problem of estimating the sample
complexity of an algorithm such as for distribution testing or noisy
binary search itself.  Proofs in this space are often sloppy with
constant factors, so the proven bound is not reflective of the true
performance.  If so, we would like to empirically estimate the sample
complexity $i$ at which the success probability $p_i$ is above a given
threshold $\tau$ (say, 90\%).  (In some cases we even know the
worst-case distribution~\cite{diakonikolas2017optimal} so the
empirical estimate is of the worst-case performance, not just the
distributional performance.)  We can run the algorithm at a given
sample complexity $i$ and check correctness, getting \textsc{Success}
with probability $p_i$.  The success probability is monotonic in $i$,
and we would like to estimate the $i^*$ where $p_i$ crosses $\tau$.
Finding $i^*$ exactly may be very hard---the success probability at
10000 and 10001 samples are likely to be almost identical---so we
would settle for \emph{some} index with $p_i \approx \tau$.

For a non-computer science example, calculating the LD50 for a substance (the dose needed to kill half of the members of a specific population) is a noisy binary search problem with error probability that skyrockets close to the true answer.

Such considerations led to the noisy binary search model
of~\cite{karp}, which we call $\NBS(\tau, \eps)$: we have $n$ coins
whose unknown probabilities $p_i \in [0, 1]$ are sorted in
nondecreasing order.  We can flip coin $i$ to see heads with
probability $p_i$.  The goal is to find any coin $i$ with nonempty
$[p_i, p_{i+1}] \cap (\tau - \eps, \tau + \eps)$.  This model subsumes
\IPNBS{} (where $p_i = \frac{1}{2} - \eps$ for $i \leq i^*$ and
$\frac{1}{2} + \eps$ otherwise) and of course regular binary search (where $p_i \in \{0, 1\}$).  Throughout this paper we will
suppose that $\tau$ is a constant bounded away from $\{0, 1\}$, $n$
grows to $\infty$, and $\eps$ and the desired failure probability
$\delta$ may be constant or may approach $0$ as $n \to \infty$.

\begin{figure}[h]\label{fig:models}
\begin{center}
\includegraphics[scale=.5]{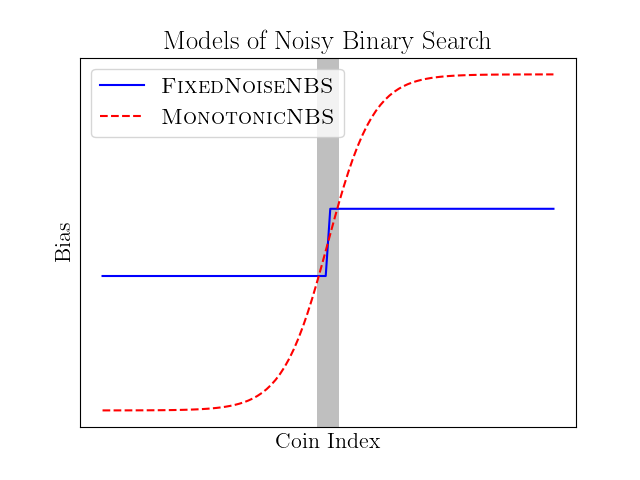}
\caption{In~\IPNBS{}, every coin is $\eps$-far from the true $i^*$
  that must be found.  We consider~\NBS{}, where many coins may be
  close to the threshold and the goal is to find some good coin (the
  gray shaded region).}
\end{center}
\end{figure}

The naive solution to \NBS{} is binary search with repetition: we do
regular binary search, but repeat each query enough times to have
$\frac{\delta}{\log n}$ failure probability if
$p_i \notin[\tau - \eps, \tau + \eps]$.  This gives sample complexity
$O(\frac{1}{\eps^2} \log n \log \frac{\log n}{\delta})$.
In~\cite{karp} it was shown that this extra $\log \log n$ term is
unnecessary, giving two algorithms that each have sample complexity
\[
  O(\frac{1}{\eps^2} \log n \log \frac{1}{\delta}).
\]

In this paper, we show how to improve this bound.  We show upper and
lower bounds that achieve the tight constant on $\log n$, and reduce
the $\log \frac{1}{\delta}$ dependence from multiplicative to
additive.  Figure~\ref{fig:boundstable} compares our result to
existing methods for~\NBS{}.

\begin{figure}
\begin{tabular}{|c|c|c|c|}
  \hline
  Algorithm&Proven query complexity& Actual constant \\
  \hline
  \hline
  Binary Search w/ Repetition & $2\frac{\tau(1-\tau)}{\eps^2} \ln n \cdot \textcolor{black}{\ln \frac{\lg n}{\delta}}$
                                   & 
  \\
  \hline
  \cite{karp} Multiplicative Weights & $4000 \frac{\max(\tau, 1-\tau)}{\eps^2} \ln n \cdot \ln \frac{1}{\delta}$ & $\approx 31$\\
  \hline
  \cite{karp} Backtracking & $476909 \frac{\max(\tau, 1-\tau)}{\eps^2} \ln n \cdot \ln \frac{1}{\delta}$ & $\approx  2000$\\
  \hline
  \textbf{\textsc{BayesianScreeningSearch}}
           &
                                                    $2\frac{\tau(1-\tau)}{\eps^2} \ln n$ &\\
    \hline
\end{tabular}
\caption{Comparison of our result to prior algorithms for \NBS{} in
  the regime of $\eps \ll \min(\tau, 1-\tau)$ and
  $\delta = 1/n^{o(1)}$, ignoring lower order terms.  The analysis
  in~\cite{karp} is not careful with constants, so we also include our
  best estimate of the actual constant after tuning constant factors
  in the algorithms.  }\label{fig:boundstable}
\end{figure}

\paragraph{On Studying Constants.}  When analyzing sublinear
algorithms, and trying to remove $\log \log n$ factors in query
complexity, constant factors really matter.  The proofs in~\cite{karp}
are not careful with constants, but the algorithms themselves
inherently lose constants.  Our best estimate is that one algorithm
``improves'' upon naive repetition by a factor of
$\frac{\ln \lg n}{31}$, and the other by $\frac{\ln \lg n}{2000}$.
Neither is an improvement for any $n$ that will ever be
practical---the better algorithm is only an improvement for
$n > 2^{e^{31}} \approx 10^{10^{11}}$.  By studying constants, we are
forced to design an algorithm that (as we shall see) gives
improvements for practical values of $n$.  We give further discussion
of the value of studying constants in Section~\ref{sec:constants}.

Noisy binary search is intimately connected to the asymmetric binary
channel, i.e., the binary channel that can choose between sending $1$
with probability $\tau - \eps$ or with probability $\tau + \eps$.  If
each $p_i \in \{\tau \pm \eps\}$, then noisy binary search needs to
reveal the $\lg n$-bit $i^*$ through such a channel; queries below
$i^*$ are $1$ with probability $\tau - \eps$ and those above $i^*$ are
$1$ with probability $\tau + \eps$.  The natural target sample
complexity is therefore $\frac{1}{C_{\tau,\eps}} \lg n$, where
$C_{\tau,\eps}$ is the \emph{information capacity} of the asymmetric
binary channel:
\begin{align}
  C_{\tau, \eps} := \max_q H( (1-q) (\tau - \eps) + q(\tau + \eps)) -
  (1-q)H(\tau-\eps) - q H(\tau+\eps)\label{eq:Ctaueps}
\end{align}
where $H(p)$ is the binary entropy function.  For
$\tau = \frac{1}{2}$, the maximum is at $q = \frac{1}{2}$ and this is
just $C_{\frac{1}{2},\eps} = 1 - H(\frac{1}{2}-\eps)$, the capacity of
the binary \emph{symmetric} channel with error probability
$\frac{1}{2} - \eps$.  For $\tau \neq \frac{1}{2}$, the information
obtained from $\tau - \eps$ and $\tau + \eps$ probability coins is not
the same, so the capacity is achieved by getting $\tau + \eps$ coins
with some probability $q$ different from $1/2$; it satisfies
$C_{\tau, \eps} \approx \frac{\eps^2}{2\tau(1-\tau)\ln 2}$ for fixed
$\tau$ as $\eps \to 0$.

\paragraph{Our results.}
Our main result is the following:

\begin{restatable}[Upper bound]{theorem}{mainresult}\label{thm:main}
  Let $0 < \tau < 1$ be a constant. Consider any parameters
  $0< \eps, \delta < 1/2$ with $0 < \eps < \min(\tau, 1-\tau)/2$.  On
  any $\NBS(\tau, \eps)$ input, the algorithm \mainalg{} uses at most
  \[
    \frac{1}{C_{\tau, \eps}} (\lg n + O(\log^{2/3} n \log^{1/3} \frac{1}{\delta} + \log \frac{1}{\delta}))
  \]
  queries and succeeds with probability
  $1-\delta$.
\end{restatable}

Unlike~\cite{BZ,bayesoptimal,DLU21,WGZW22,GX}, our results apply
to~\NBS{}, not just \IPNBS{}, so they do not restrict the value of
$p_i$ and handle $\tau \neq \frac{1}{2}$.  Unlike~\cite{karp}, we
achieve good constant factors, high-probability results, and a better
scaling with the target $\tau$.  In particular,~\cite{karp} scales
multiplicatively rather than additively with
$O(\log \frac{1}{\delta})$; and it uses a reduction that incurs a
constant-factor \emph{loss} for targets $\tau \neq \frac{1}{2}$, while
Theorem~\ref{thm:main} scales with $\Theta(\tau(1-\tau))$ so
\emph{improves} for $\tau \neq \frac{1}{2}$.

Using Shannon's strong converse theorem, we show that the dependence
on $n$ is tight: for $\eps \gg n^{-1/4}$, any algorithm must sometimes
use $(1 - o(1)) \frac{1}{C_{\tau, \eps}} \lg n$ queries; in fact, it
must use this many queries with nearly $1-\delta$ probability.

\begin{restatable}[Strong converse]{theorem}{ipnbslowerboundrevised}\label{thm:ipnbslowerboundrevised}
  Any $\NBS(\tau, \eps)$ algorithm that succeeds with $1 - \delta$
  probability on inputs with all $p_i \in \{\tau \pm \eps\}$ must have
  at least a $1 - \delta - O(\frac{1}{\gamma^2 n \eps^4})$ chance of using at least
  \[
    (1-\gamma)\frac{1}{C_{\tau, \eps}} \lg n
  \]
  queries, for any $\gamma > 0$.
\end{restatable}
For $\tau = \frac{1}{2}$, this is also a lower bound for $\IPNBS$.
Thus Theorem~\ref{thm:ipnbslowerboundrevised} gives a new
\emph{worst-case} lower bound for $\IPNBS$, which is a
$\frac{1}{1-\delta}$ factor larger than the lower bound for
\emph{expected} query complexity achieved in prior
work~\cite{BZ,bayesoptimal,DLU21,GX}.

For $\tau \neq \frac{1}{2}$, our results are the first ones connecting
noisy binary search to $C_{\tau,\eps}$, the information capacity of
the binary asymmetric channel.

\paragraph{Our results: expected queries.}
For constant $\delta$, one can get a better bound for the expected
number of queries in a simple way: only run the algorithm with
probability $1 - (1 - \frac{1}{\log n})\delta$, and otherwise output
the wrong answer from zero queries.  This saves essentially a
$1-\delta$ factor in queries, which for constant $\delta$ is
nontrivial:

\begin{restatable}[Upper bound: expected queries]{corollary}{mainresults}\label{cor:expect}
  Under the same conditions as Theorem~\ref{thm:main} and for any $\NBS(\tau, \eps)$ input, algorithm
  \mainalgsilly{} uses
  \[
    \frac{1-\delta}{C_{\tau, \eps}} (\lg n + O(\log^{2/3} n \log^{1/3} \frac{\log n}{\delta} + \log \frac{1}{\delta}))
  \]
  queries in expectation and succeeds with probability $1-\delta$.
\end{restatable}

This $1-\delta$ savings is essentially the best possible.  Our strong
converse (Theorem~\ref{thm:ipnbslowerboundrevised}) already implies
this, if $\eps \gg n^{-1/4}$; but using Fano's inequality, the
optimality is true in general:
\begin{restatable}[Weak converse]{theorem}{ipnbscorollary}\label{thm:ipnbscorollary}
  Any $\NBS(\tau, \eps)$ algorithm that succeeds with $1 - \delta$
  probability on inputs with all $p_i \in \{\tau \pm \eps\}$ must use
  \[
    (1-\delta)\frac{\lg (n - 2) - 1}{C_{\tau, \eps}} 
  \]
  queries in expectation.
\end{restatable}
Theorem~\ref{thm:ipnbscorollary} was essentially shown in~\cite{BZ},
which proved the $\tau = \frac{1}{2}$ case (by giving hardness for \IPNBS{}).

\paragraph{Our results: experiments.}
In Section~\ref{sec:experiments} we compare our approach to naive
repetition and the~\cite{karp} algorithms.  We find, for $n \geq 10^3$
and $\eps = .1$, that our approach outperforms naive repetition,
which outperforms both~\cite{karp} algorithms.  For $n = 10^9$, our
approach uses $2.3\times$ fewer samples than naive repetition.

\subsection{Algorithm Overview}
We now describe our noisy binary search
algorithm in the case of $\tau = \frac{1}{2}$ and
$\delta > 1/n^{o(1)}$.

\paragraph{Bayesian start.} The natural choice for a ``hard'' instance is
when $p_i \in \{\tau \pm \eps\}$, so the algorithm must find the
transition location $i^*$, and information theoretic arguments show
$\frac{1}{C_{\tau, \eps}} \lg n$ queries are necessary.  To avoid
losing a constant factor in sample complexity, the algorithm
essentially must spend most of its time running the Bayesian
algorithm.  This algorithm starts with a uniform prior over which
interval crosses $\tau$, makes the maximally informative query,
updates its posterior, and repeats.  When $\tau = \frac{1}{2}$, the
maximally informative query is the median under the posterior, and the
Bayesian update is to multiply intervals on one side of the query by
$1 + 2 \eps$ and the other side by $1 - 2 \eps$.  This algorithm,
\bayesalg, is given in Algorithm~\ref{alg:bayes}; the algorithm for
general $\tau$ is given in Section~\ref{sec:keylemma}.

As a technical side note, the discrete nature of the problem
introduces a bit of subtlety.  Note that \NBS{} flips coins $i$ but
 returns an \emph{interval} between coins that should be good:

\begin{definition}
  We say that an interval $[i, i+1]$ is \emph{$(\tau, \eps)$-good} if
  $[p_i, p_{i+1}] \cap (\tau-\eps, \tau+\eps)$ is nonempty.
\end{definition}

Precisely, our version of the Bayesian algorithm is as follows: we
start with a uniform prior over intervals.  The median of our
posterior can be viewed as a fractional coin, and we flip the nearest
actual coin but update our posterior as if we flipped the fractional
coin.  So, for example, suppose the median is $4.7$
($.7 * w(5) + \sum_{i=1}^4 w(i) = .5$).  We flip coin 5, and if it
comes out $0$, that suggests the true threshold is probably above $5$.
We then scale up our posterior on all intervals above $5$ by
$1 + 2 \eps$; scale down intervals below $4$ by $1-2\eps$; and scale
the weight on interval $[4, 5]$ by $.3 (1 + 2\eps) + .7(1-2\eps)$.
This new posterior is still a distribution that sums to $1$.

\begin{algorithm}
  \caption{Bayesian learner in $\tau = \frac{1}{2}$ case. Flips $M$ coins and returns $M$ intervals.}\label{alg:bayes}
  \hspace*{\algorithmicindent} \textbf{Input} A set of $n$ queryable coins, update size $\eps$, number of steps $M$. \\
 \hspace*{\algorithmicindent} \textbf{Output} A list of $M$ intervals queried.
  \begin{algorithmic}[1]
    \Procedure{\bayesalg}{$\text{coins}, \varepsilon, M$}
    \State $n \gets \abs{\text{coins}}$
    \State $w_1 \gets \text{uniform}([n-1])$ \Comment{Prior distribution over intervals}
    \State $L \gets \{\}$
    \For{$i \in [M]$}
    \State $j_i \gets $ median interval of $w_i$
    \State $x_i \gets $ either $j_i$ or $j_i+1$, whichever is closer to the median
    \State append $j_i$ to $L$
    \State $y_i \gets $ flip coin $x_i$ \Comment{$1$ with probability $p_{x_i}$}
    \State $w_{i+1}(x) \gets 
    \begin{cases} 
      w_i(x)\cdot (1 - 2\eps(-1)^{y_i})  &\mbox{if } x < j_i \\ 
      w_i(x)\cdot (1 + 2\eps(-1)^{y_i}) &\mbox{if } x > j_i\\
      \text{remainder so $w_{i+1}$ sums to $1$}  &\mbox{if } x = j_i
    \end{cases}
    $
    \EndFor
    \State \Return $L$
    \EndProcedure
  \end{algorithmic}
\end{algorithm}

\paragraph{Using the result.}
After running the Bayesian algorithm for most of our query budget, we
need to output an answer. The question becomes: how can we take the
transcript of the Bayesian algorithm and extract a useful worst-case
frequentist guarantee?  We need the algorithm to work for \emph{all}
monotonic $p$, which can have values very different than
$\tau \pm \eps$.

In the prior work achieving tight constants
for~\IPNBS{}~\cite{BZ,DLU21}, because the $p_i$ are guaranteed to be
$\frac{1}{2} \pm \eps$, the analysis can show that the weight of the
single ``good'' interval grows in expectation at each step.  By a
Hoeffding bound, after the desired number of iterations the ``good''
interval has more weight than every other interval combined, so it can
be easily selected.  But that property is not true for the more
general $p_i$ of~\NBS{}: if many $p_i$ are $\frac{1}{2} \pm 0.6 \eps$,
the Bayesian algorithm will wander somewhat too slowly through these
samples without growing any single interval by the desired amount.

However, in such cases the Bayesian algorithm is spending a lot of
time among good intervals.  
This holds in general.  Our key lemma shows that, if we run $\bayesalg$ for $1 + O(\gamma)$
times the information theoretic bound $\frac{1}{C_{\tau,\eps}} \lg n$,
a $\gamma$ fraction of the intervals it visits are
$(\tau, \eps)$-good:

\begin{restatable}[Bayesian performance]{lemma}{keylemma}
  \label{lem:keylemma}
  Consider any $0 < \eps, \tau, \delta, \gamma < 1$ with
  $\gamma \leq \frac{1}{7}$, $\eps < \min(\tau,1-\tau)/2$, and let $L$ be the list of intervals
  returned by \bayesalg, when run for
  \[
     \frac{1 + O(\gamma)}{C_{\tau, \eps}} \cdot \left(\lg n + O(\sqrt{\log n \log \frac{1}{\delta}} + \log \frac{1}{\delta})\right)
  \]
  iterations on an $\NBS$ instance.  With probability $1-\delta$, at least a $\gamma$ fraction of the intervals
  in $L$ are $(\tau, \eps)$-good.
\end{restatable}

By considering the $\gamma$-quantiles of the returned list, we reduce
$n$ to $\frac{1}{\gamma}$.  We can now run a less efficient noisy
binary search algorithm on this small subproblem.  There are some
complications, as the solution to the new noisy binary search could
correspond to a larger interval than two adjacent coins. To deal with
this, we run \bayesalg{} with $\eps^\prime = (1 - o(1))\eps$, which
lets us test our candidate answers.

\paragraph{Technical comparison of techniques.}
How we leverage the bayesian learner is the main technical difference
between our upper bound and that of prior work \cite{karp,BZ,DLU21}.
As described above, the situation is rather simpler for~\IPNBS.
For~\NBS{}, \cite{karp} instead used conservative updates in their
multiplicative weights algorithm: rather than the true Bayesian update
$1 \pm 2\eps$, it multiplies by about $1 \pm \frac{3}{5} \eps$.  This
necessarily loses a constant factor, but ensures that either the
median interval queried or the last interval queried is good.  This
property is not true for the true Bayesian algorithm with sharp
constant.

\subsection{Related Work}

The~\IPNBS{} version of noisy binary search, where
$\tau = \frac{1}{2}$ and $p_i \in \{\frac{1}{2} \pm \eps\}$, was posed by
Burnashev and Zigangirov~\cite{BZ}, who showed how to achieve
\[
  \frac{1}{C_{\frac{1}{2},\eps}} \left(\lg n + \lg \frac{1}{\delta} + \lg \frac{1 + 2\eps}{1-2\eps}\right)
\]
\emph{expected} queries (in Russian; see~\cite{WGZW22} for an English
proof).  Essentially the same~\cite{BZ} algorithm for~\IPNBS{} was
rediscovered in~\cite{bayesoptimal}.  Some bugs with the~\cite{bayesoptimal} proof were
discovered and fixed in \cite{DLU21}, as well as an analysis of a
variant of the algorithm for \emph{worst-case} sample complexity
\[
  \frac{1}{C_{\frac{1}{2},\eps}} \left(\lg n + O(\sqrt{\log n \log \frac{1}{\delta}}) + O(\log \frac{1}{\delta})\right).
\]

For $1 \ll \log \frac{1}{\delta} \ll \log n$, Gu and
Xu~\cite{GX} showed black-box improvements for other $\delta$. If
$\delta$ is constant, they output $\perp$ with probability
$\delta - \frac{1}{\log n}$, and otherwise run the~\cite{DLU21} algorithm
with $\delta' = \frac{1}{\log n}$.  On the other hand, for
$\delta = n^{-\Omega(1)}$, repeatedly running~\cite{DLU21} with
$\delta' = \frac{1}{\log n}$ and checking the result gives
improvements:
\[
  (1 + o(1)) \left(\frac{1-\delta}{C_{\frac{1}{2},\eps}} \lg n +  \frac{\log \frac{1}{\delta}}{\eps \log \frac{1+2\eps}{1-2\eps}}\right)
\]
For $\eps \ll 1$, this is a factor 2 improvement on the constant
factor on $\log \frac{1}{\delta}$.  Moreover,~\cite{GX} shows that
this bound is sharp in both $n$ and $\delta$.

Our version of noisy binary search,~\NBS{}, was first posed by Karp
and Kleinberg~\cite{karp}.  They gave two algorithms, based on
recursive backtracking and multiplicative weights respectively, that
take $O(\frac{1}{\eps^2}\log n)$ queries for constant $\delta$, which
they showed is within constant factors of optimal for constant
$\tau, \delta$.  Unfortunately, the constant factors make both
algorithms worse than the naive repetition algorithm for any
reasonable $n$ (see Figure~\ref{fig:boundstable} and
Section~\ref{sec:experiments}).

\paragraph{Other models.}

There are many different variations for noisy binary search (see
\cite{pelcsurvey} for a survey of older work on the subject).
Emamjomeh-Zadeh, Kempe, and Singhal~\cite{noisygraphbinarysearch}
solve an extension of~\IPNBS{} from the line to graphs.  This result was
improved and simplied by Dereniowski, Tiegel, Uznański and Wolleb-Graf \cite{DTUW20}, 
which was in later improved and simplified by Dereniowski, Łukasiewicz, and Uznański
\cite{DLU21}.  Nowak developed a different generalization of~\IPNBS{}
to general hypothesis classes \cite{nowak}.  Waeber, Frazier, and
Henderson~\cite{noisybisectionsearch} investigates a continuous
variant of~\IPNBS{}, where the target is a point in the real interval
$[0,1]$, and show that the Bayesian algorithm converges geometrically
(the ideal convergence up to constant factors).

To our knowledge,~\cite{karp} is the only previous work that handles a
setting like~\NBS{} where the ``true'' coin may be indistinguishable
from nearby coins, and the goal is just to find a sufficiently good
answer.

\paragraph{Applications.}
Noisy binary search is also used as a subroutine in other
algorithms. For instance in \cite{nbstesting} it is used for group
testing, and in Crume~\cite{gitbisect} as a replacement for git-bisect
under unreliable tests.  Both implementations were based on the
multiplicative weights algorithm of Karp and Kleinberg~\cite{karp}.

\subsection{Why constants?}\label{sec:constants}
There is a tendency in theoretical computer science to regard constant
factors as unimportant.  But theorists care about constants in many
situations, such as approximation ratios or rates of codes, and we
believe that the query complexity of sublinear algorithms is another
situation where they should be considered.

In general, the arguments for ignoring constants in time complexity
hold with much less force for query complexity. The constant for time
complexity is highly dependent on the machine architecture, which
changes over time (e.g., the relative cost of addition and
multiplication). Moreover, these hardware improvements mitigate the
cost of poor constants. But the number of queries is a mathematical
value, and the cost of queries (which may be, e.g., blood tests or
running a giant test suite) does not clearly decrease with time.

The question should be: does theoretical study of constant factors
lead to algorithmic insights necessary for more practical algorithms?
Our paper shows that it does. By considering constants, we are forced
to find a more efficient way of translating the Bayesian algorithm
into one with frequentist guarantees (via
Lemma~\ref{lem:keylemma}). The constants lost in the previous attempt
at this (in~\cite{karp}) mean that it is worse than the naive method
until $n > 10^{10^{11}}$.

It should not be surprising that a simple method that loses an
$O(\log \log n)$ factor can beat an algorithm that loses ``only''
constants, for all practical values of $n$. The study of leading
constants is a lens by which we found a new algorithm that actually
outperforms the naive method for reasonable values of $n$ (namely
$n > 1000$).

\section{Detailed Proof Sketch for Upper Bound}

\subsection{Key Lemma on Bayesian Learner.}

For this proof overview, we focus on the case of $\delta > n^{-o(1)}$ and target $\tau = \frac{1}{2}$,
where \textsc{BayesLearn} queries the median of the posterior at each
stage, and
\[
  C_{\tau, \eps} = 1 - H(\frac{1}{2} + \eps) = (\frac{1}{2} + \eps) \lg (1 + 2 \eps) + (\frac{1}{2} - \eps) \lg (1 - 2 \eps)
  \approx \frac{2\eps^2}{\ln 2}.
\]

We give an overview of the proof of our key lemma in this case:

\keylemma*

Let $a$ be the ``best answer'', an interval that straddles the bias $\frac{1}{2}$.

The algorithm keeps track of a distribution $w$ on $[n-1]$; at each
step $i$, it queries the median of the current distribution $w_i$,
then multiplies the density on one side by $1 + 2\eps$ and the other
by $1-2\eps$ to form $w_{i+1}$.  We analyze the algorithm by looking
at $\lg w(a)$.

At each step, the interval $j$ we choose is either good (a valid answer) or bad (invalid).  If it
is \emph{bad}, suppose the sampled coin $x$ has probability
$p_x \geq \frac{1}{2} + \eps$.  Then $x$ is above $a$, so
$w(a)$ multiplies by $1+2\eps$ with probability $p_x$, and $1-2\eps$
with probability $1-p_x$.  Hence:
\[
  \E[\lg w_{i+1}(a) - \lg w_i(a)] = p_x \lg (1 + 2\eps) + (1-p_x) \lg (1 - 2\eps) \geq  C_{\tau,\eps}.
\]
The case of $p_x \leq \frac{1}{2} - \eps$ is symmetric, giving the
same bound.  So every bad interval we select increases $\lg w(a)$ by
$C_{\tau, \eps}$ in expectation.

On the other hand, if the interval we select is \emph{good}, $\lg w(a)$
may decrease in expectation.  For example, if we query coin $a$ and $\sum_{i=1}^{a-1} w(i) = \frac{1}{2}$, 
we could have
\[
  \E[\lg w_{i+1}(a) - \lg w_i(a)] = \frac{1}{2}\lg (1 - 2 \eps) + \frac{1}{2} \lg (1 + 2 \eps) \approx -\frac{2\eps^2}{\ln 2} \approx -C_{\tau,\eps}
\]
It turns out this is essentially the worst case, and in general the
expected decrease in $\lg w(a)$ is no more than $5C_{\tau,\eps}$ for
any $\eps < \frac{1}{2}\min(\tau, 1-\tau)$.
As a result, the potential function
\[
  \lg w_i(a) - \gamma C_{\tau, \eps} \cdot (\text{\# intervals chosen}) + 6 C_{\tau,\eps} \cdot (\text{\# good intervals chosen})
\]
increases by at least $(1-\gamma)C_{\tau,\eps}$ in expectation in each
step $i$, regardless of where the median is in that step.  This
potential function starts at $-\lg (n-1)$, so after
$M = (1 + 2\gamma) \frac{1}{C_{\tau,\eps}}\lg n$ steps it is at least
$\Theta(\gamma) \lg n$ in expectation.  An Azuma-Hoeffding bound
shows that the value concentrates about this expectation, and in particular will be positive with
$1 - \delta$ probability.  If so, since $\lg w_i(a) \leq 0$ always,
we have
\[
  6 \cdot (\text{\# good intervals chosen}) - \gamma (\text{\# intervals chosen}) \geq 0,
\]
and hence a $\frac{\gamma}{6}$ fraction of chosen intervals are
good.

This proves the key lemma: after
$(1 + O(\gamma)) \frac{1}{C_{\tau,\eps}}\log n$ steps of \bayesalg, a
$\gamma$ fraction of coins flipped are good with decent probability.

\paragraph{Targets $\tau \neq \frac{1}{2}$.}  When $\tau \neq \frac{1}{2}$, the
maximum-information query is no longer the median coin, but a slightly
different quantile $\frac{1}{2} \pm O(\frac{\eps}{\tau(1-\tau)})$, and
the Bayesian updates use more complicated factors.  This choice is
still capacity-achieving on bad intervals, i.e., the expected
``information gain'' is
$\E[\lg w_{i}(a) - \lg w_{i+1}(a)] \geq C_{\tau,\eps}$, and on good intervals
the expected information loss is still at most
$5C_{\tau,\eps}$, so the proof structure works unchanged.

\subsection{Rest of Upper Bound}
Recall that in this overview we assume
$\log \frac{1}{\delta} \ll \log n$.  By Lemma~\ref{lem:keylemma}, if
we take all $\{ \gamma, 2\gamma, \dots, \floor{\frac{1}{\gamma}}\gamma\}$ quantiles of the list returned by
\bayesalg, run with parameter $\eps^\prime = \eps(1 - \alpha)$ (where $\alpha$ is introduced so we can later test the bias of each coin), we get
a size-$\frac{1}{\gamma}$ list containing at least one
$\eps^\prime$-good interval.  This $\eps^\prime$ has
$C_{\tau,\eps^\prime} = (1 - O(\alpha))C_{\tau,\eps}$.  For any
$\gamma$, we can just flip all of these coins
$O(\frac{1}{\alpha^2\eps^2} \log \frac{1}{\gamma \delta})$ times to find an $\eps$-good one.  This would give sample
complexity
\begin{align}\label{eq:simpleralg}
  \underbrace{(1 + O(\gamma)) (1 + O(\alpha)) \frac{1}{C_{\tau,\eps}} \left(\lg n + O(\sqrt{\log n \log \frac{1}{\delta}})\right)}_{\text{\bayesalg, Lemma~\ref{lem:keylemma}}} + \underbrace{O(\frac{1}{\gamma} \cdot \frac{1}{\alpha^2 \eps^2} \log \frac{1}{\gamma \delta})}_{\text{Testing quantiles}}
\end{align}
which, by setting $\gamma$ and $\alpha$ to $(\frac{\log \frac{1}{\delta}}{\log n})^{1/4}$, gives sample complexity
\[
  (1 + O(\frac{\log \frac{1}{\delta}}{\log n})^{1/4})\frac{1}{C_{\tau,\eps}} \lg n.
\]

This is the desired sharp bound, within $(1 + o(1))$ of optimal.  One
can do slightly better: the second stage is itself a noisy binary
search question on $O(1/\gamma)$ coins, so by applying the algorithm
recursively with $\gamma' = O(1)$ we can solve it on the
size-$O(1/\gamma)$ list in
$O(\frac{1}{(1 - \alpha)\capacity} \log \frac{1}{\gamma \delta} +
\frac{1}{\alpha^2\eps^2}\log \frac{1}{\gamma \delta})$
queries. As we recurse on a much smaller list, the samples used are all lower order and we do not need to recurse more than once. However, the answer to the recursive call might not be a
valid answer to the original problem. Regardless, one of the endpoints
of the return call must be a valid answer, which we can test for.
By optimizing the parameters, this improves the sample complexity to
\[
  (1 + O(\frac{\log \frac{1}{\delta}}{\log n})^{1/3})\frac{1}{C_{\tau,\eps}} \lg n,
\]
giving Theorem~\ref{thm:main}.

\section{Proof of Lemma~\ref{lem:keylemma}}\label{sec:keylemma}

\newcommand{\pp}{{\prime\prime}}
\newcommand{\ppp}{{\prime\prime\prime}}
\newcommand{\z}{z}
\newcommand{\zexp}{\ensuremath{2^{\frac{H(\tau - \eps) - H(\tau + \eps)}{2\eps}}}}
\newcommand{\m}{m}
\newcommand{\mexp}{\frac{1 - (1 - \tau - \eps)(1+\z)}{2\eps(1 + \z)}}
\newcommand{\y}{y}
\newcommand{\yexp}{\frac{1}{1 + \z}}
\newcommand{\D}{\ensuremath{d}}
\newcommand{\dexp}{\frac{1}{1 + \frac{1 - \frac{\eps}{\tau}}{1 - \frac{\eps}{1 - \tau}}}}
\newcommand{\constant}{6 \capacity}
\newcommand{\expdoo}{\frac{1 - \tau - \eps}{1 - \tau - (2q-1)\eps}}
\newcommand{\expdol}{\frac{1 - \tau + \eps}{1 - \tau - (2q-1)\eps}}
\newcommand{\expdlo}{\frac{\tau + \eps}{\tau + (2q-1)\eps}}
\newcommand{\expdll}{\frac{\tau - \eps}{\tau + (2q-1)\eps}}

\subsection{Definitions}

Let $\{l, \dots, r\}$ be the set of good intervals. Let $a$ be the maximum $i \in [n-1]$ such that $p_i \leq \tau$.

Recall that we defined
\begin{align}
	 C_{\tau, \eps} &= \max_q H( (1-q) (\tau - \eps) + q(\tau + \eps)) - (1-q)H(\tau-\eps) - q H(\tau+\eps)
\end{align}

We also define the following functions:

\begin{align}
	 W(x) &= \sum_{i \in [x]} w(i)\\
	\Phi(w,L) &= \lg w(a) + \constant(|\{x \in L | x \in [l,r]\}| - \gamma|L|)
\end{align}

$\capacity$ is the capacity of a $(\tau,\eps)$-BAC. We let $q$ satisfy the equation
\[
  q = \argmax_x  H((1-x)(\tau - \eps) + x(\tau+\eps)) - (1-x)H(\tau - \eps) - xH(\tau
  + \eps),
\]
which expresses the shared information between a sent and received
message through a $(\tau,\eps)$-BAC. (See \ref{eq:cdef}, \ref{eq:qdef} for explicit formulas for $\capacity, q$) If our prior were true---so the
coins really were $\tau \pm \eps$---we would like to flip a
$\tau+\eps$ coin with probability $q$.  This is achieved by selecting
the $q$-quantile of our posterior, which is above the true
threshold with probability $q$.  If $\tau = \frac{1}{2}$, $q = \frac{1}{2}$ and we
query the median; in general, we query the
$q = \frac{1}{2} \pm O(\frac{\eps}{\tau(1-\tau)})$ quantile.

$\Phi$ is a potential function that we will be analyzing.

We also define:

\begin{align}
	d_{0,0} &= \frac{1 - \tau - \eps}{1 - \tau - (2q-1)\eps} & \label{eq:doo}\\
	d_{0,1} &= \frac{1 - \tau + \eps}{1 - \tau - (2q-1)\eps} & \label{eq:dol}\\
	d_{1,0} &= \frac{\tau + \eps}{\tau + (2q-1)\eps} & \label{eq:dlo}\\
	d_{1,1} &= \frac{\tau - \eps}{\tau + (2q-1)\eps} & \label{eq:dll}
\end{align}
for brevity. In terms of $\bayesalg$ we can think of $d_{x,y}$ as
``the multiplicative effect of a flip resulting in $x$
($1 = \text{Heads}, 0 = \text{Tails}$) on the density of an interval
on side $y$ ($1 = \text{Right}, 0 = \text{Left}$) of the flipped
coin.''  When $\tau = \frac{1}{2}$, $d_{x,y} = 1 - 2 \eps (-1)^{x \oplus y}$.

\begin{algorithm}
\caption{Acts as a Bayesian learner for $M$ iterations, returns a list of all the chosen intervals. Expressions for the $d_{x,y}$ values are given in \eqref{eq:doo}, \eqref{eq:dol}, \eqref{eq:dlo}, \eqref{eq:dll}}
\begin{algorithmic}[1]
\Procedure{getIntervalFromQuantile}{$w,q$}
	\State $i \gets \min {i \in [n]} $ s.t. $W(i) \geq q$
\EndProcedure
\Procedure{roundIntervalToCoin}{$i,w,q$}
	\State \Return $i$ \textbf{if} $\frac{q - W(i-1)}{w(i)} \leq q$ \textbf{else} $i + 1$
\EndProcedure
\Procedure{\bayesalg}{$\{c_i\}_{i=1}^n, n, \tau, \varepsilon, M$}
	\State $w_1 \gets \text{uniform}([n-1])$
	\State Define $q$ as in~\eqref{eq:qdef} \Comment{The quantile we choose}
	\State $L \gets \{\}$
	\For{$i \in [M]$}
		\State $j_i \gets \textsc{getIntervalFromQuantile}(w_i, q)$ \Comment{The chosen interval}
		\State $x_i \gets \textsc{roundIntervalToCoin}(j_i, w_i, q)$ \Comment{The index of the coin we are going to flip}
		\State append $j_i$ to $L$
		\State $y_i \gets \textsc{flip}(c_{x_i}) $
		\State $w_{i+1} \gets 
			\begin{cases} 
				w_i(x)d_{y_i,0} &\mbox{if } x \in \{1,\dots,j_i-1\} \\ 
				d_{y_i,0}(q - W_i(j_i - 1)) + d_{y_i,1}(W_i(j_i) - q)  &\mbox{if } x = j_i\\
				w(x)d_{y_i,1} &\mbox{if } x \in \{j_i + 1, \dots, n - 1\}
			\end{cases}$
	\EndFor
	\Return $L$
\EndProcedure
\end{algorithmic}
\end{algorithm}

\subsection{Analysis}

\begin{lemma}
\label{Martingale}
 $\E[\Phi_{t+1} - \Phi_{t}|y_t, y_{t-1}, \dots, y_1] \geq (1 - O(\gamma))C_{\tau,\eps}$.
\end{lemma}

\begin{proof}
	$\Phi$ is given by the sum of equations (\ref{eq:term1}) and (\ref{eq:term2}).

	\begin{equation}\label{eq:term1}
		\constant(|\{j \in L | j \in [l,r]\}| - \gamma|L|)
	\end{equation}

	\begin{equation}\label{eq:term2}
		\lg w(a)
	\end{equation}

        Recall that in the $t$th round, $j_t$ is the interval chosen, and $x_t$ is index of the coin flipped. Let $p$ be the probability $c_{x_t}$ lands heads.

        \paragraph{Bad queries.}  Suppose $j_t \notin [l, r]$.
        If $j_t > r$, then $p \geq \tau + \eps$ and the expected change in
        \eqref{eq:term2} is
        \[
          p \lg d_{1,0} + (1-p)\lg d_{0,0}
        \]
        The first $\log$ is positive and the second $\log$ is
        negative, so this expression is minimized at $p = \tau + \eps$, at which point some
        computation (Lemma~\ref{lem:bayesgain}) shows that it equals
        $\capacity$.  Similarly, if $x_t< l$ then $p \leq \tau - \eps$ and
        the expected change is
        \[
          p \lg d_{1,1} + (1-p)\lg d_{0,1}
        \]
        which is also at least $\capacity$ by Lemma~\ref{lem:bayesgain}.

        As $j_t \not\in [l,r]$, the change in (\ref{eq:term1}) is $-\gamma \cdot \constant$.

	   Therefore in this case the expected change in $\Phi$ is at least $(1 - 6\gamma)\capacity$.

        \paragraph{Good queries.}  Suppose $j_t \in [l, r]$. The
        change in \eqref{eq:term1} is now $\constant(1 - \gamma)$.
        But how much can~\eqref{eq:term2} decrease in expectation?

	   Suppose that $j_t \neq a$. Then the expected change is either 

	   \[
		p \lg d_{1,0} + (1 - p)\lg d_{0,0}
	   \]

	   with $p \geq \tau$, or 

	   \[
		p \lg d_{1,1} + (1 - p)\lg d_{0,1}
	   \]

	   with $p \leq \tau$.

	   As $d_{1,0} = \expdlo \geq \expdoo = d_{0,0}$ and $d_{1,1} = \expdll \leq \expdol = d_{0,1}$, both these expressions are minimized when $p = \tau$.

	   So the expected change in $\eqref{eq:term2}$ is lower bounded by:

	   \begin{equation}
	   \min\left(\tau \lg d_{1,0} + (1 - \tau) \lg d_{0,0}, \tau \lg d_{1,1} + (1 - \tau) \lg d_{0,1}\right).
	   \end{equation}

	   We note that

	   \begin{align*}
		   \tau \lg d_{1,0} + (1 - \tau) \lg d_{0,0} &= (\tau + \eps)\lg d_{1,0} + (1 - \tau - \eps) \lg d_{0,0} - \eps \lg d_{1,0} + \eps \lg d_{0,0}\\
		&= \capacity - \eps \lg d_{1,0} + \eps \lg d_{0,0} \\
		&\geq \capacity - 3\eps(\frac{\eps}{\tau} + \frac{\eps}{1 - \tau}) & \text{Lemma~\ref{lem:logdbounds}}  \\
		&= \capacity - \frac{3\eps^2}{\tau(1 - \tau)}\\
		&\geq\capacity - (6\ln 2 )\capacity & \text{Lemma~\ref{lem:capacitybounds}}\\
		&\geq -5\capacity
	   \end{align*}
	a symmetric argument for lower bounding $\tau \lg d_{1,1} + (1 - \tau) \lg d_{0,1}$ holds. Therefore, the change in \eqref{eq:term2} is lower bounded by $-5\capacity$.

	   Now suppose that $j_t = a$.

	   Then the expected change in \eqref{eq:term2} is:

	   \[
		p \lg(d_{1,0} k + d_{1,1} (1-k)) + (1 - p)\lg(d_{0,0} k + d_{0,1} (1 - k))
	   \]

	   for some $k \in [0,1]$. 

	   If $k \leq q$ then we flip $a$ so $p \leq \tau$. $d_{0,0} k + d_{0,1} (1 - k) \geq d_{0,0} q + d_{0,1} (1-q) = 1$. Also $d_{1,0} k + d_{1,1} (1-k) \leq d_{1,0} q + d_{1,1} (1-q) = 1$. Therefore, this expression is minimized when $p = \tau$. By symmetry, when $k > q$ this expression is also minimized when $p = \tau$.

	   So the expected change in $\eqref{eq:term2}$ is lower bounded by

	   \[
		\tau \lg(d_{1,0} k + d_{1,1} (1-k)) + (1 - \tau)\lg(d_{0,0} k + d_{0,1} (1 - k))
	   \]
	   for some $k \in [0,1]$.

	   Taking the derivative with respect to $k$, we get

	   \[
		\tau\frac{d_{1,0} - d_{1,1}}{d_{1,1} + (d_{1,0} - d_{1,1})k} + (1 - \tau)\frac{d_{0,0} - d_{0,1} }{d_{0,1} + (d_{0,0} - d_{0,1})k}
	   \]

	   As $d_{1,1} < d_{1,0}$ and $d_{0,1} > d_{0,0}$, $\tau\frac{d_{1,0} - d_{1,1}}{d_{1,1} + (d_{1,0} - d_{1,1})k} > 0 > (1 - \tau)\frac{d_{0,0} - d_{0,1} }{d_{0,1} + (d_{0,0} - d_{0,1})k}$. We note that as $k$ increases, $\tau\frac{d_{1,0} - d_{1,1}}{d_{1,1} + (d_{1,0} - d_{1,1})k}$ decreases in magnitude, while $(1 - \tau)\frac{d_{0,0} - d_{0,1} }{d_{0,1} + (d_{0,0} - d_{0,1})k}$ increases in magnitude. Therefore, the minimum value of the above expression is achieved when $k=0$ or $k = 1$. 

	   So the expected change in $\eqref{eq:term2}$ is lower bounded by
	   \[
		\min(\tau \lg d_{1,0} + (1 - \tau) \lg d_{0,0}, \tau \lg d_{1,1} + (1 - \tau) \lg d_{0,1})
	   \]

	   which is the same expression which we lower bounded for the $j_t \neq a$ case. 

	   Combining these two cases, when the we are querying a good interval, the expected change is lower bounded by $\constant(1 - \gamma) - 5 \capacity = (1 - 6\gamma)\capacity$

 	Therefore $\E[\Phi_{t+1} - \Phi_{t}|y_t, y_{t-1}, \dots, y_1] \geq (1 - O(\gamma))C_{\tau,\eps}$.

\end{proof}

Now we prove our key lemma.

\keylemma*

\begin{proof}
	Recall that $\Phi$ is given by the sum of equations (\ref{eq:term1}) and (\ref{eq:term2}).
	\begin{equation*}\tag{\ref{eq:term1}}
		\constant(|\{j \in L | j \in [l,r]\}| - \gamma|L|)
	\end{equation*}

	\begin{equation*}\tag{\ref{eq:term2}}
		\lg w(a)
	\end{equation*}
	\item
	\paragraph{Reduction to $\Phi > 0$.} First note that $\Phi_1 = -\lg(n - 1)$, as initially $L$ is empty so (\ref{eq:term1}) is $0$, and we initialize $w$ as uniform so $w(a) = \frac{1}{n-1}$.
	Next note that if $\eqref{eq:term1} > 0$, then
	\begin{align*}
		\constant(|\{x \in L| x \in [l,r]\}| - \gamma|L|) &> 0\\
		|\{x \in L| x \in [l,r]\}|  &> \gamma|L|
	\end{align*}
	So there are strictly more than $\gamma|L|$ good intervals in $L$. 
	Next note that \eqref{eq:term2} is $\leq 0$ always, so $\Phi > 0 \implies $\eqref{eq:term1} $> 0$.
	So it suffices to show that with $\delta$ failure probability $\Phi_{t+1} > 0$, where $t  = \frac{1 + O(\gamma)}{C_{\tau, \eps}} (\lg n + O(\sqrt{\log n \log \frac{1}{\delta}} + \log \frac{1}{\delta}))$

	\paragraph{Establishing a submartingale.}

	Note by a stochastic domination argument, we can consider the worst case where all coins sampled have bias in $[\tau - \eps, \tau + \eps]$.

	Let $\sigma_i^2$ be the variance of the difference random variables $\Phi_{i+1} - \Phi_i$, then we note that $\sigma_i^2$ is a Bernoulli random variable with parameter $p \in [\tau - \eps, \tau + \eps]$, that is scaled by at most a $\max(\lg d_{1,0} - \lg d_{0,0}, \lg d_{0,1} - \lg d_{1,1}) \lesssim \frac{\eps}{\tau(1-\tau)}$ factor, therefore

	\begin{align*}
		\sigma_i^2 \lesssim p(1-p) \left(\frac{\eps}{\tau(1-\tau)}\right)^2 \lesssim \tau(1-\tau)\left(\frac{\eps}{\tau(1-\tau)}\right)^2 = \frac{\eps^2}{\tau(1-\tau)} \lesssim \capacity
	\end{align*}

	Where we use the fact that $\eps \leq \frac{\min\left(\tau,1-\tau\right)}{2}$.

	Therefore $\sigma_i^2 \lesssim \capacity$.

	\paragraph{Freedman's inequality.}

	For brevity let $g = (1 - O(\gamma))\capacity$, the lower bound given in Lemma~\ref{Martingale}.

	\begin{align*}
		\Pr[\Phi_{t+1} \leq 0] &= \Pr[\Phi_{t+1} - \Phi_1 \leq -\Phi_1] \\
		 &= \Pr[\Phi_{t+1} - gt - \Phi_1 \leq -gt - \Phi_1] \\
		&\leq \exp\left(-\frac{2(-gt - \Phi_1)^2}{\sum_{i=1}^t\sigma_i^2 + O(\frac{\eps}{\tau(1-\tau)})(gt + \Phi_1)}\right) && \text{Freedman's when } gt \geq -\Phi_1\\
		&\leq \exp(- \frac{O(1)}{t \capacity} \cdot (-gt - \Phi_1)^2) \\
		&= \exp(- \frac{O(1)}{t \capacity} \cdot (g^2t^2 + 2gt\Phi_1 + \Phi_1^2)) \\
	\end{align*}

	Bounding this expression by $\delta$, we get

	\begin{align}
		\nonumber
		\exp\left(- \frac{O(1)}{t \capacity} \cdot (g^2t^2 + 2gt\Phi_1 + \Phi_1^2)\right)&\leq \delta \\ \nonumber
		g^2t^2 + 2g\Phi_1t + \Phi_1^2 &\geq \log_{O(1)}(1/\delta)t\capacity\\
		g^2t^2 + (2g\Phi_1 - \log_{O(1)}(1/\delta)\capacity)t + \Phi_1^2 &\geq 0 \label{eq1}
	\end{align}

	(\ref{eq1}) is a quadratic with respect to $t$, and has a positive leading coefficient. Applying the quadratic formula, if we set
	\begin{align}
		\nonumber
		t &\geq \frac{-(2g\Phi_1 - \log_{O(1)}(1/\delta)\capacity) + \sqrt{(2g\Phi_1 - \log_{O(1)}(1/\delta)\capacity)^2 - 4g^2\Phi_1^2}}{2g^2}\\ \nonumber
		&= \frac{-(2g\Phi_1 - \log_{O(1)}(1/\delta)\capacity) + \sqrt{(\log_{O(1)}(1/\delta)\capacity)^2 - 4g\Phi_1 \log_{O(1)}(1/\delta)\capacity}}{2g^2}\\
		&= \frac{-\Phi_1}{g} + \frac{\capacity\log_{O(1)}(1/\delta)}{2g^2} + \frac{\sqrt{(\log_{O(1)}(1/\delta)\capacity)^2 - 4g\Phi_1 \log_{O(1)}(1/\delta)\capacity}}{2g^2} \label{eq2}
	\end{align}

	then $(\ref{eq1})$ holds.

	As $\Phi_1 = -\lg(n - 1), g = (1-O(\gamma))\capacity$ we get that \eqref{eq2} is

 \[
    \frac{1}{1 - O(\gamma)} \cdot \frac{1}{C_{\tau, \eps}} (\lg n + O(\sqrt{\log n \log \frac{1}{\delta}} + \log \frac{1}{\delta}))
  \]
  
  As $\frac{1}{1 - O(\gamma)}$ is $1 + O(\gamma)$, our lemma holds.
\end{proof}

\section{Algorithm and Analysis}

\begin{algorithm}
\caption{Noisy Binary Search. It recurses at most once.}
\begin{algorithmic}[1]

\Procedure{ReductionToGamma}{$\{c_i\}_{i=1}^n, n, \tau, \varepsilon, \delta, \gamma$}
	\State $L \gets \bayesalg(\{c_i\}_{i=1}^n, n, \tau, \varepsilon, \gammaiterations)$
	\State $R \gets \{\}$
	\For{$i \in [\floor{\frac{|L|}{\ceil{\gamma |L|}}}]$}
		\State append $L_{\ceil{\gamma |L|}i}$ to $R$
	\EndFor
	\Return $\textsc{removeDuplicates}(R)$
\EndProcedure

\Procedure{\NBSName}{$\{c_i\}_{i=1}^n, n, \tau, \varepsilon, \delta, \gamma = \frac{1}{7\lg (n)}$}
	\State $\varepsilon^\prime = \varepsilon \cdot \max(1 - \sqrt[3]{\log_n(1/\delta)}, \frac{2}{3})$
	\State $R \gets \textsc{ReductionToGamma}(\{c_i\}_{i=1}^n, n, \tau, \varepsilon^\prime, \delta/3, \frac{1}{3\lg (n)})$
	\If{$|R| > 7$}
		\State $R \gets [1] + R + [n]$ \Comment{Pad $R$ with the extremes of the initial problem}.
		\State $i \gets \textsc{\NBSName}(\{c_{R_i}\}_{i=1}^{|R|}, |R|, \tau, \varepsilon^\prime, \delta/3, \frac{1}{7})$
		\State $\hat{p}_{R_i+1} \gets \text{estimate } p_{R_i + 1} \text{ up to } \pm \frac{\eps - \eps^\prime}{2} \text{ error with } \delta/3 \text{ f.p.}$ 
		\If{$\hat{p}_{R_i+1} > \tau - \eps + \frac{\eps - \eps^\prime}{2}$}
			\State \Return $R_i$
		\Else
			\State \Return $R_{i+1} - 1$
		\EndIf
	\Else
		\For{$x \in R$}
		\State $\hat{p}_{x+1} \gets \text{estimate } p_{x + 1} \text{ up to } \pm \frac{\eps - \eps^\prime}{2} \text{ error with } \delta/18 \text{ f.p.}$ 
		\If{$\hat{p}_{x+1} > \tau - \eps + \frac{\eps - \eps^\prime}{2}$}
			\State \Return $x$
		\EndIf
		\EndFor
	\EndIf
\EndProcedure

\end{algorithmic}
\end{algorithm}

\begin{algorithm}
\caption{Noisy Binary Search that gets the optimal expected queries.}
\begin{algorithmic}[1]
\Procedure{\SillyNBSName}{$\{c_i\}_{i=1}^n, n, \tau, \varepsilon, \delta$}
\State \Return $\begin{cases} \textsc{Random}([n - 1])& \text{w.p. } \delta - \delta/\lg n \\ \textsc{\NBSName}(\{c_i\}_{i=1}^n, n, \tau, \varepsilon, \frac{\delta}{\ln(n)}) & \text{otherwise}\end{cases}$
\EndProcedure
\end{algorithmic}
\end{algorithm}

\mainresult*

\begin{proof}
	\newcommand{\pinterval}{(\tau - \eps^\prime, \tau + \eps^\prime)}
	\newcommand{\einterval}{(\tau - \eps, \tau + \eps)}
\item
	\paragraph{Correctness.}
	Suppose that we run \textsc{\NBSName} on a \NBS{} instance with parameters $\{c_i\}_{i=1}^n,n,\tau, \varepsilon,\delta$.

	Assume that all probabilistic stages succeed, meaning that \textsc{ReductionToGamma}, \textsc{BayesianScreeningSearch}, and our coin bias estimation all succeed. By a union bound, this occurs with probability $\geq 1 - \delta$.

	As we pick every $\gamma|L|$th coin from $L$ and $L$ contains at least $\ceil{\gamma|L|}$ $\eps^\prime$-good intervals, $R$ contains at least one $\eps^\prime$-good interval. Suppose that $|R| \leq 7$ and that $R_i$ is the first $\eps^\prime$-good interval in $R$.

	Then for all $j \in \{1,\dots,i-1\}$, either $R_j$ is an $\eps$-good interval or it is not. If it is, then we have nothing to worry about outputting it. If it is not, then $p_{R_j + 1} \leq \tau - \eps$ (as if $p_{R_j + 1} \geq \tau + \eps$ then $R_i$ is not $\eps$-good), so $\hat{p}_{R_j + 1} \leq \tau - \eps + \frac{\eps - \eps^\prime}{2}$. So we do not output any not $\eps$-good interval before $R_i$. 

	Once we reach $R_i$, $p_{R_i + 1} > \tau - \eps^\prime$, so $\hat{p}_{R_i + 1} > \tau - \eps^\prime - \frac{\eps - \eps^\prime}{2} = \tau - \eps + \frac{\eps - \eps^\prime}{2}$ and we output $R_i$.

	Now suppose that $|R| > 7$. As we recursively run \textsc{\NBSName} with $\gamma = 1/7$, we note that for the $R$ in the recursive call $R^\prime$, $|R^\prime| = \floor{\frac{|L|}{\ceil{\gamma|L|}}} \leq \floor{\frac{1}{\gamma}} = 7$, so $|R^\prime| \leq 7$. By our work above, this means that the recursive call returns $i$ such that $[p_{R_i}, p_{R_{i + 1}}] \cap \pinterval \neq \emptyset$.

	Either $R_i$ or $R_{i + 1} - 1$ is $\eps^\prime$-good, as if $p_{R_i+1} \leq \tau - \eps^\prime$ and $p_{R_{i+1}} - 1 \geq \tau + \eps^\prime$ then $R$ must not contain any good intervals. The same logic as for the $|R| \leq 7$ case holds, and we have shown correctness.

	\paragraph{Number of samples.}
	Next we analyze the sample budget.

	Suppose that we run $\textsc{\NBSName}$ with $\gamma = 1/7$.

	The \textsc{ReductionToGamma} call takes

	\begin{align*}
		\frac{1 + O(\gamma)}{C_{\tau,\eps^\prime}} \left(\lg  n + O(\sqrt{\log n \log \frac{1}{\delta}} + \log \frac{1}{\delta})\right) &= \frac{1}{\capacity} O(\log n + \log \frac{1}{\delta})
	\end{align*}

	samples.

	As we have $\gamma = 1/7$, $|R| \leq 7$ and we go through the second branch.

	Then the bias estimation takes $O(\frac{\tau(1-\tau)\log \frac{1}{\delta}}{(\eps - (1 - \sqrt[3]{\log_n \frac{1}{\delta}})\eps)^2}) = O(\frac{\tau(1-\tau)\log \frac{1}{\delta}}{(\eps\sqrt[3]{\log_n \frac{1}{\delta}})^2}) = O(\frac{\log^{2/3}n \log^{1/3} \frac{1}{\delta}}{\capacity}) $ samples, for overall $\frac{1}{\capacity} O(\log n + \log \frac{1}{\delta})$ samples.

	Now consider the case $\gamma = \frac{1}{7\lg n}$, and suppose that $1 - \sqrt[3]{\log_n(1/\delta)} \geq 2/3$.

	 \textsc{ReductionToGamma} takes, with $\gamma = O(1/\log(n)), \eps^\prime = \eps*(1 - \sqrt[3]{\log_n(1/\delta)})$:

	\begin{align*}
	&\frac{1 + O(\gamma)}{C_{\tau,\eps^\prime}} \left(\lg  n + O(\sqrt{\log n \log \frac{1}{\delta}} + \log \frac{1}{\delta})\right) \\
	&= \frac{1 + O(\frac{1}{\log n})}{C_{\tau,\eps^\prime}} \cdot \left(\lg  n + O(\sqrt{\log n \log \frac{1}{\delta}} + \log \frac{1}{\delta})\right)\\
																														  &= \frac{1}{C_{\tau,\eps^\prime}} \cdot \left(\lg  n + O(\sqrt{\log n \log \frac{1}{\delta}} + \log \frac{1}{\delta})\right) \\
																														  &= \frac{1}{(1 - O(\sqrt[3]{\log_n (1/\delta)})\capacity} \cdot \left(\lg  n + O(\sqrt{\log n \log \frac{1}{\delta}} + \log \frac{1}{\delta})\right) & \text{(Lemma~\ref{lem:capacityepsprime})}
	\end{align*}

	samples, which is $\frac{1}{\capacity} \cdot \left(\lg  n + O(\log ^{2/3}n\log ^{1/3}\frac{1}{\delta} + \log \frac{1}{\delta} )\right)$.

	If $|R| \leq 7$ we take the second branch and take $O(\frac{\log^{2/3}n \log^{1/3} \frac{1}{\delta}}{\capacity})$ more samples, which meets our bound.

	 If $|R| > 7$ we take the first branch and recurse with $\gamma = 1/7$ and $n^\prime = O(\log n)$, for $\frac{1}{\capacity} O(\log \log n + \log \frac{1}{\delta})$ samples.

	 As established previously, the bias estimation takes $O(\frac{\log^{2/3}n \log^{1/3} \frac{1}{\delta}}{\capacity})$ samples.

	 For overall
	  \[
	    \frac{1}{C_{\tau, \eps}} (\lg n + O(\log^{2/3} n \log^{1/3} \frac{1}{\delta} + \log \frac{1}{\delta}))
	  \]

	samples.

	In the case $1 - \sqrt[3]{\log_n(1/\delta)} < 2/3$, the $\frac{O(\log \frac{1}{\delta})}{\capacity}$ term dominates the rest, and the bound holds.
\end{proof}

\mainresults*

\begin{proof}
	The failure probability of \textsc{\SillyNBSName} is $\leq \delta - \delta/\lg n + (1 - \delta + \delta/\lg n) \delta/\lg n = \delta - \delta^2/\lg n + \delta^2/\lg^2 n \leq \delta$.

	We use $0$ samples with probability $\delta - \delta/\lg n$, and the expression in Theorem~\ref{thm:main} with $\delta^\prime = \delta/\lg n$, with probability $1 - \delta + \delta/\lg n$.
\end{proof}

\section{Lower Bounds}

\begin{lemma}
\label{red:fromcom}
	Given any algorithm $\mathcal{A}$ which solves $NBS$ for parameters $(\tau,\eps)$ with sample budget $m$ and failure probability $\delta$, there exists a protocol that communicates over a discrete memoryless channel with capacity $\capacity$ with rate $R = \frac{\lg(n - 1)}{m}$ with failure probability $\delta$.
\end{lemma}

\begin{proof}
	Let channel $\mathcal{C}$ be a $(\tau, \eps)$-BAC with shared randomness and perfect feedback.\\
	Binary asymmetric channels are discrete memoryless channels, and so neither feedback nor shared randomness change its channel capacity. \cite{shannonfeedback} Therefore, the capacity of $\mathcal{C}$ is $\capacity$. \\
	Suppose we have agents $A$ and $B$, and $A$ wishes to communicate a message $x^* \in [n-1]$  to $B$ over $\mathcal{C}$. Also assume without loss of generality $\mathcal{A}$ always flips a coin exactly $m$ times.

	Both $A$ and $B$ can run an identical copy of $\mathcal{A}$, as we have shared randomness. When the algorithm flips a coin $x^*$, $A$ sends $B$ $1$ if $x^* < x$, and $0$ otherwise. Then if $x^* < x$, $B$ receives $1$ with probability $\tau + \eps$, and $0$ with probability $1 - \tau - \eps$. If $x^* \geq x$, then $B$ receives $1$ with probability $\tau - \eps$ and $0$ with probability $1 - \tau + \eps$. With perfect feedback we can have $A$  receive the same value that $B$ received. Note that to $B$ this is just an NBS problem with parameters $n,\tau,\eps$, and so it successfully recovers $x^*$ with probability $\geq 1 - \delta$. The rate of our simulated code is $R = \frac{\lg(n-1)}{m}$, and so the lemma holds.
\end{proof}

Now we can use lower bounds from information theory. 

\begin{lemma}[Shannon's Strong Converse Theorem]
\label{shannonconverse}
	Over any discrete memoryless channel, for $R > C$
	$$P_e \geq 1 - \frac{K_1}{n(R-C)^2} - \exp(-K_2n(R - C))$$

	where $P_e$ is the probability of error, $K_1,K_2$ are positive constants which depend on the channel, $n$ is the input alphabet size, $R$ is the rate of information, and $C$ is the channel capacity
	\cite{gallager}
\end{lemma}

\ipnbslowerboundrevised*

\begin{proof}
	Let $\alpha = \frac{1}{1 + \frac{K}{\capacity\sqrt{\beta(n-1)}}}$, for constant $K$ to be determined later. Suppose that $\mathcal{A}$ uses at most $\alpha \frac{\lg(n-1)}{\capacity}$ samples with probability at least $\delta + \beta$. Let $\mathcal{A}^\prime$ be the algorithm that runs $\mathcal{A}$, but outputs a random answer if $\mathcal{A}$ uses more than $\alpha \frac{\lg(n-1)}{\capacity}$ samples. $\mathcal{A}^\prime$ fails only whenever $\mathcal{A}$ fails or uses more than $\alpha \frac{\lg(n-1)}{\capacity}$ samples, so by a union bound $\mathcal{A}^\prime$ has a failure probability of at most $1 - \beta$.

	By Lemma~\ref{red:fromcom} we can construct a protocol over a discrete memoryless channel with capacity $\capacity$ that communicates at rate $R = \frac{\lg(n-1)}{\frac{\alpha \lg(n-1)}{\capacity}} = \frac{\capacity}{\alpha} = \frac{C}{\alpha}$ with a failure probability of at most $1 - \beta$.

	By Lemma~\ref{shannonconverse} we have that
	\begin{align*}
	1 - \beta &\geq 1 - \frac{K_1}{(n-1)(R-C)^2} - \exp(-K_2(n-1)(R - C))\\
	&= 1 - \frac{K_1}{(n-1)((1/\alpha - 1) C)^2} - \exp(-(n-1)K_2((1/\alpha - 1) C))\\
	&= 1 - \frac{K_1\beta}{K} - \exp(-\sqrt{\frac{n-1}{\beta}}K_2K)\\
	&\geq 1 - \frac{\beta}{2} & \hspace{-1in}\text{sufficiently large } K \text{ and } n\\
	\end{align*}

	which is a contradiction. Therefore with probability at least $1 - \delta - \beta$ $\mathcal{A}$ uses $\frac{1}{1 + \frac{K}{\sqrt{\beta(n-1)}(\capacity)}} \frac{1}{\capacity} \lg n$ samples.

\end{proof}

We can also lower bound the expected number of samples.

\ipnbscorollary*

\begin{proof}
	Suppose we have algorithm $\mathcal{A}$ which uses $q$ samples in expectation to solve NBS with $\delta$ failure probability. By Lemma \ref{red:fromcom} and Fano's inequality
	$\lg (n-1) - q\capacity \leq 1 + \delta\lg (n - 2) \implies (1-\delta)\lg (n - 2) -1 \leq q\capacity \implies (1-\delta)\frac{\lg (n-2)-1}{\capacity} \leq q$.
\end{proof}

\section{Experiments}\label{sec:experiments}

\paragraph{Applying NBS.}

To demonstrate the practicality of \textsc{\NBSName} we compare it to
standard binary search with repetition (\textsc{NaiveNBS}) and the two algorithms
of~\cite{karp} (\textsc{KKBacktracking} and \textsc{KKMultiplicativeWeights}).

To fairly compare between these algorithms, we can't just use the descriptions given in~\cite{karp}, as the constants used in analysis are not optimized. We leverage {\textsc{\NBSName}} to address this. We tweak the listed algorithms so they take a sample budget as input which they allocate among all their stages. To estimate how large a budget is needed for algorithm $\mathcal{A}$ to perform well on distribution $\mathcal{D}$, we run \textsc{BayesianScreeningSearch} where when the $i$th coin is flipped we run $\mathcal{A}$ on some input drawn from $\mathcal{D}$, and return $1$ if $\mathcal{A}$ succeeds and $0$ if $\mathcal{A}$ fails. By setting $\tau = .8,.9$ and $\eps = .05$, we get upper and lower bounds for how many samples is needed to get $\delta=.15$ failure probability.

\paragraph{Experiments.}

We compare results on 4 different problem distributions: \textsc{Standard}, \textsc{Biased}, \textsc{Lopsided}, and \textsc{Wide}. 

\begin{itemize}
	\item[\textsc{Standard}] $p_i \in \{\tau - \eps, \tau + \eps\}$, $\tau = \frac{1}{2}, \eps = .1$, the transition interval chosen uniformly at random..
	\item [\textsc{Biased}] $p_i \in \{\tau - \eps, \tau + \eps\}, \tau = \frac{3}{4}, \eps = .1$, the transition interval chosen uniformly at random.
	\item[\textsc{Lopsided}] $p_i \in \{\tau - .6\eps, \tau + \eps\}$, $\tau = \frac{1}{2}, \eps = .1$, the transition interval chosen uniformly at random..
	\item[\textsc{Wide}] we choose an interval (uniformly at random) of size $10\ln n$ that linearly interpolates between $\tau - \eps$ and $\tau + \eps$, and set the rest to be $p_i \in \{\tau -\eps, \tau + \eps\}, \tau = \frac{1}{2}, \eps = .1$.
\end{itemize}

\paragraph{Results.}

We remark that $\textsc{KKBacktracking}$ performed markedly worse than
the other algorithms, and so is not included in the figures. For
reference, for $\textsc{Standard}, N = 1000$ $\textsc{KKBacktracking}$
required $m > 2.9 \times 10^{6}$ samples, while the other algorithms
need $m < 6000$ samples (see Figure~\ref{fig:results}).

\begin{figure}
\includegraphics[scale=0.5]{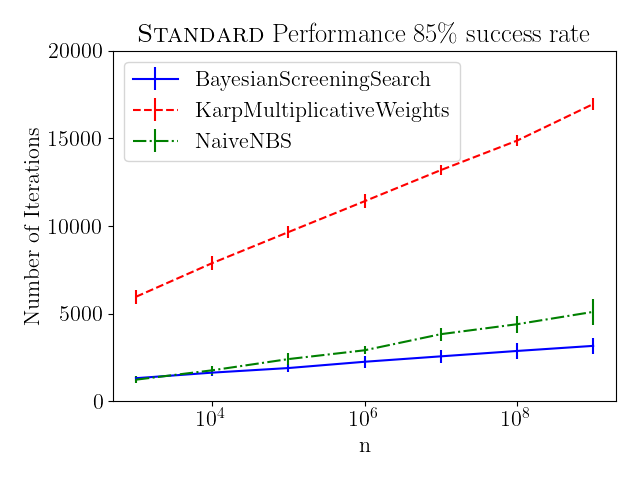}
\includegraphics[scale=0.5]{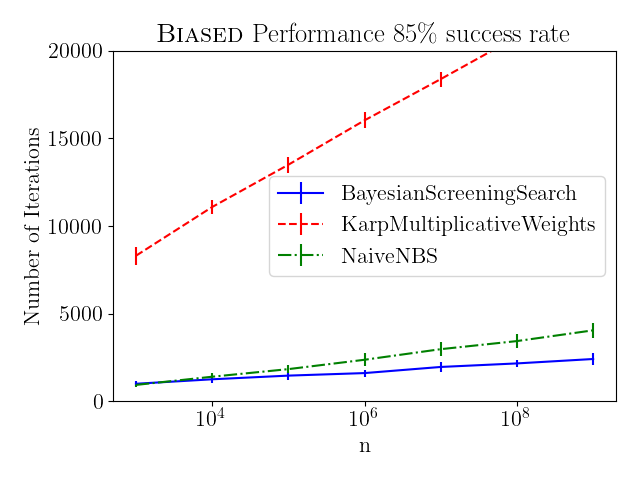}
\includegraphics[scale=0.5]{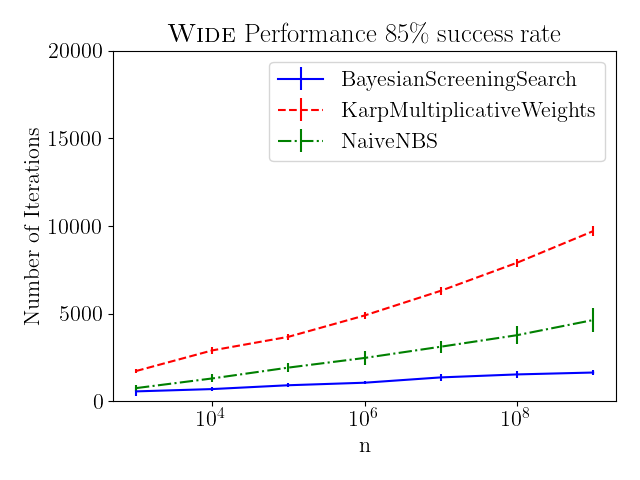}
\includegraphics[scale=0.5]{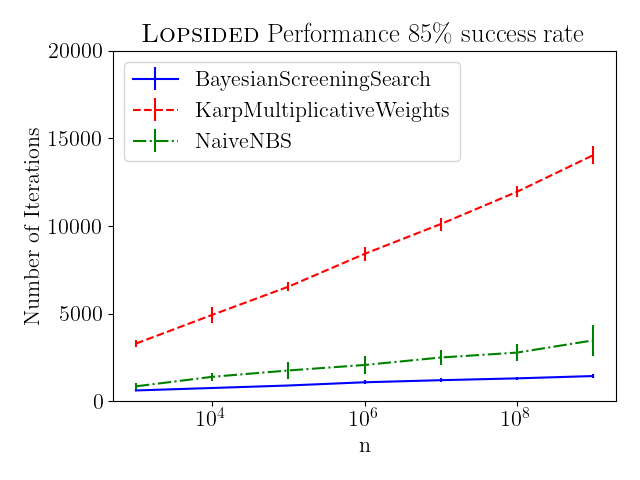}
\caption{Performance of various \NBS{} algorithms for listed distributions.}\label{fig:results}
\end{figure}

We find that~\textsc{KKMultiplicativeWeights} is outperformed by
\textsc{NaiveNBS} on all of these distributions. In contrast, {\textsc{\NBSName}} outperforms~\textsc{NaiveNBS} for $n > 10^3$.  

When $\tau \neq \frac{1}{2}$ the difference between {\textsc{\NBSName}}, \textsc{NaiveNBS} and the \cite{karp} algorithms increases. This is in line with our theory, as the first two perform better when $\tau$ is further from $\frac{1}{2}$, while the \cite{karp} algorithms reduce to the case $\tau=\frac{1}{2}$, losing a constant factor.

\paragraph{\textsc{\NBSName} variant.}

Since \textsc{\NBSName} has a large set of parameters to tune, we use a simpler variant.
We run \textsc{\NBSName} as normal but we use $\eps$ to update instead
of $\eps^\prime$ when running \bayesalg. This can be shown to satisfy
Theorem~\ref{thm:main} as well.

To see that this gets the same constant, suppose we are trying to run \bayesalg\text{ } with parameters $n,\tau,(1 - \alpha)\eps$, but in the actual algorithm we use $\eps$ instead of $(1 - \alpha)\eps$. Then the expected gain in $\Phi$ when the chosen interval is above the good range is 
\begin{align*}
	(\tau + (1-\alpha)\eps)\lg d_{1,0} + (1 - \tau - (1-\alpha)\eps)\lg d_{0,0} &= \capacity - \alpha\eps * \lg \frac{\expdlo}{\expdoo} \\ &\geq \capacity - \alpha * O(\frac{\eps^2}{\tau(1-\tau)})\\ &= (1 - O(\alpha))\capacity\\
\end{align*}

The case when the chosen interval is below the good range of intervals is symmetric. When the chosen interval is good, we can see that the loss is -$\Theta(\capacity)$ using the same work as in Lemma~\ref{lem:keylemma}. So setting $\alpha = \sqrt[3]{\log_n \frac{1}{\delta}}$ we get the same constant as \textsc{\NBSName}.

We set $\gamma = \frac{1}{\ln^2 n}$, and addition, when we recurse on $1/\gamma$ elements we run $\textsc{NaiveNBS}$ to find the $2$ coins out of the possible $7$ to test.

\paragraph{Implementation details.}

To run noisy binary search on  each algorithm we need to modify each algorithm to ``solve NBS with a given sample budget,'' instead of ``given a NBS instance solve it with as few samples as possible.'' In this section we discuss implementation decisions made.

To efficiently implement \textsc{\NBSName} we use a lazily initialized segment tree, to perform any operations on $w$ in $O(\log n)$ time. When running the algorithm with a sample budget $b$, $\frac{\lg n}{\capacity}$ was allocated to the 
\textsc{ReductionToGamma} call, $O(\frac{b\log \log n}{\capacity})$ was allocated to the Recursive \textsc{BayesianScreeningSearch} call, and $O(\frac{1}{\capacity})$ was allocated to the final bias estimation call, and the remaining budget is split among the three stages evenly.
 
To efficiently implement \textsc{KKMultiplicativeWeights} we also use a lazily initialized segment tree. When running the algorithm with a sample budget $b$, we determine the maximum number of iterations of update step we can do with this budget and perform this many steps. (In terms of the original paper, we noisy binary search on $T(n)$).

For \textsc{NaiveNBS}, for a sample budget $b$ we allocate the number of flips evenly to each of the $\lg n$ steps of the algorithm.

For \textsc{KKBacktracking} we make no modifications, as the algorithm is structured to eventually output an answer but we do not have limits on the number of samples for any stage. So when running the algorithm with sample budget $b$ we cause it to fail if the algorithm uses more than $b$ samples.

\section{Future Work}

One interesting topic of research is \emph{instance-dependent} noisy
binary search. If an instance is much nicer than the worst case, say
every coin has bias $\frac{1}{2} \pm \alpha$ for $\alpha \gg \eps$, we
would hope to get a $O(\frac{\log n}{\alpha^2})$ dependence, which
\textsc{\NBSName} does not get. One could use an adaptive coin bias
estimator to get some adaptivity, but the constants gotten from this
will likely not be good.

Another open problem is attenuating the lower order terms in the upper
bound for NBS. For realistic $n$, lower order terms such as
$\sqrt{\log n}$, or even $\log \log n$ are not negligible compared to
$\log n$, and influences the practical application of
\textsc{\NBSName}, as seen in the experimental results where we spent
$28\%$ of our samples on the ``lower order'' recursive calls.

One conjectural algorithm for noisy binary search would be: run
{\bayesalg} for $(1 + O(\gamma)) OPT$ steps, then output the median of
the last $\gamma OPT$ intervals chosen.  This interpolates between the
overall median (which loses a constant factor) and the final interval
(which has a large probability of failure), and avoids the
inefficiency of recursive calls.

\printbibliography
\appendix
\section{Computations}
This section gives the proof of some approximations used in the body
of the paper.

We give explicit formulas for some functions used in this paper

\begin{align}
	 \z &= \zexp \\
	 C_{\tau,\eps} &= \lg(\z+1) +  \frac{\tau - \epsilon}{2\epsilon}H(\tau + \epsilon) - \frac{\tau+\epsilon}{2\epsilon}H(\tau-\epsilon)\label{eq:cdef}\\
	 q &= \frac{\left(1 - \tau + \eps  \right) - \frac{1}{1 + \z}}{2\eps}\label{eq:qdef}
\end{align}

\begin{lemma}
	\label{lem:bayesgain}
	\begin{equation*}
	C_{\tau,\eps} = (\tau + \epsilon)\lg(\frac{\tau+\epsilon}{\tau + (2q-1)\eps}) + (1 - \tau - \epsilon)\lg(\frac{1 - \tau - \epsilon}{1 - \tau - (2q-1)\eps})
	\end{equation*}
	and
	\begin{equation*}
	C_{\tau,\eps}= (\tau - \epsilon)\lg(\frac{\tau-\epsilon}{\tau + (2q-1)\eps}) + (1 - \tau + \epsilon)\lg(\frac{1 - \tau + \epsilon}{1 - \tau - (2q-1)\eps})
	\end{equation*}
\end{lemma}

\begin{proof}
  $\tau + (2q - 1)\eps = \tau + \eps - \frac{1}{1+\z} + (1 - \tau -
  \eps) = 1 - \frac{1}{1+\z}$ and similarly
  $1 - \tau - (2q - 1)\eps = 1 - \tau - \eps + \frac{1}{1+\z} + (1 -
  \tau - \eps) = \frac{1}{1+\z}$.  Then the first result is

  \begin{align*}
    &\qquad(\tau + \epsilon)\lg(\frac{\tau+\epsilon}{1 - \frac{1}{1+\z}}) + (1 - \tau - \epsilon)\lg(\frac{1 - \tau - \epsilon}{\frac{1}{1+\z}})\\
    &= -H(\tau + \epsilon) +(\tau+\epsilon)\lg(\frac{\z + 1}{\z}) +(1 - \tau - \epsilon)\lg(\z+1) \\
    &= -H(\tau + \epsilon) + \lg(\z+1) + (\tau+\epsilon)\lg(\frac{1}{\z})\\
    &= -H(\tau + \epsilon) + \lg(\z+1) - (\tau+\epsilon)(\frac{H(\tau-\epsilon) - H(\tau+\epsilon)}{2\epsilon})\\
    &= \lg(\z+1) +  \frac{\tau - \epsilon}{2\epsilon}H(\tau + \epsilon) - \frac{\tau+\epsilon}{2\epsilon}H(\tau-\epsilon)\\
    &= C_{\tau, \epsilon}.
  \end{align*}
  and similarly,
  \begin{align*}
    &\qquad (\tau - \epsilon)\lg(\frac{\tau-\epsilon}{1 - \frac{1}{1+\z}}) + (1 - \tau + \epsilon)\lg(\frac{1 - \tau + \epsilon}{\frac{1}{1+\z}}) \\
    &= -H(\tau - \epsilon) +(\tau-\epsilon)\lg(\frac{\z + 1}{\z}) +(1 - \tau + \epsilon)\lg(\z+1) \\
    &= -H(\tau - \epsilon) + \lg(\z+1) + (\tau-\epsilon)\lg(\frac{1}{\z})\\
    &= -H(\tau - \epsilon) + \lg(\z+1) - (\tau-\epsilon)(\frac{H(\tau-\epsilon) - H(\tau+\epsilon)}{2\epsilon})\\
    &= \lg(\z+1) +  \frac{\tau - \epsilon}{2\epsilon}H(\tau + \epsilon) - \frac{\tau+\epsilon}{2\epsilon}H(\tau-\epsilon)\\
    &= C_{\tau, \epsilon}.
  \end{align*}
\end{proof}

\begin{lemma}
\label{lem:capacitybounds}
  For $\eps \leq \frac{1}{2}\min(\tau, 1-\tau)$,
  \[
    \frac{1}{2 \ln 2} \frac{\eps^2}{\tau(1-\tau)} \leq \capacity \leq
    \frac{1}{\ln 2} \frac{\eps^2}{\tau(1-\tau)}.
  \]
\end{lemma}

\begin{proof}
  Note that $H^{(2)}(x) = \frac{-1}{(\ln 2) x(1-x)}$ and
  $H^{(4)}(x) = - \frac{2}{\ln 2} \left(\frac{1}{(1-x)^3} +
    \frac{1}{x^3}\right) \leq 0$ for all $x$.

  For the lower bound, by the definition of channel capacity,
  $C_{\tau, \eps}$ is at least the information gained when choosing the median, i.e.,
  \begin{align*}
    C_{\tau,\eps} &= \max_q H(\tau + (2q-1) \eps) - (q H(\tau + \eps) + (1-q)H(\tau - \eps))\\
                  &\geq H(\tau) - \frac{1}{2}  \left(H(\tau + \eps) + H(\tau - \eps)\right).
  \end{align*}
  If we Taylor expand about $\tau$, the odd powers of $\eps$ cancel,
  leaving:
  \[
    \frac{1}{2}(H(\tau + \eps) + H(\tau + \eps)) = H(\tau) + \frac{\eps^2}{2}H^{(2)}(\tau) + \frac{\eps^4}{2\cdot 4!} \left(H^{(4)}(\tau + a_1) + H^{(4)}(\tau - a_2)\right)
  \]
  for some $0 \leq a_i \leq \eps$.  Since $H^{(4)}(x) \leq 0$ for all $x$, this gives
  \[
    C_{\tau,\eps} \geq -\frac{\eps^2}{2}H^{(2)}(\tau) = \frac{\eps^2}{2 \ln 2 \tau(1-\tau)}.
  \]

  For the upper bound, the condition on $\eps$ implies that
  $H^{(2)}(x) \geq - \frac{2}{\ln 2 \tau(1-\tau)}$ for all
  $x \in [\tau - \eps, \tau + \eps]$.  Then Taylor's theorem gives,
  for some values $a, b, c \in [-\eps, \eps]$, that  
  \begin{align*}
    C_{\tau, \eps} &= H(\tau + (2q-1) \eps) - (q H(\tau + \eps) + (1-q)H(\tau - \eps))\\
                   &= H(\tau) + (2q-1) \eps H'(\tau) + \frac{1}{2} (2q-1)^2 \eps^2 H^{(2)}(\tau + a)\\
    &\quad
      - \left(H(\tau) + q \eps H'(\tau) - (1-q) \eps H'(\tau) + q \frac{\eps^2}{2}H^{(2)}(\tau + b) + (1-q) \frac{\eps^2}{2}H^{(2)}(\tau + c)\right)\\
                   &= \frac{\eps^2}{2} \left((2q-1)^2 H^{(2)}(\tau + a) - q H^{(2)}(\tau + b) - (1-q) H^{(2)}(\tau + c)\right)\\
                   &\leq \frac{\eps^2}{ \ln 2 \cdot \tau(1-\tau)}.
  \end{align*}
\end{proof}

\begin{lemma}
  For $0 < \eps \leq \frac{1}{2}\min(\tau, 1-\tau)$,
  \[
    |q - \frac{1}{2}| \leq \frac{2\eps}{\tau(1-\tau)}.
  \]
\end{lemma}

\begin{proof}
  Recall that $q \in [0, 1]$ is chosen to maximize the expected information gain:
  \begin{equation}
    \label{eq:maxcap}
    H((1-q)(\tau - \eps) + q(\tau+\eps)) - (1-q) H(\tau - \eps) - q H(\tau + \eps).
  \end{equation}
  Setting the derivative of this to zero, we get
  \begin{align}
    \nonumber
    0 &= 2\eps H^\prime(\tau + (2q-1)\eps) + H(\tau - \eps) - H(\tau+\eps)\\
    H^\prime(\tau + (2q - 1)\eps)  &=  \frac{H(\tau+\eps) - H(\tau-\eps)}{2\eps} \label{eq:objective}
  \end{align}
  As $H^\prime$ is strictly decreasing, $m$ is the unique  solution to \eqref{eq:objective}.

  Observe that
  \begin{align*}
    H(\tau + \eps) - H(\tau - \eps) = 2 \eps H'(\tau) + \frac{1}{3} \eps^3 (H^{(3)}(\tau + a) + H^{(3)}(\tau + b))
  \end{align*}
  for some $a, b \in [\tau-\eps,\tau+\eps]$.  Since
  \[
    H^{(3)}(x) = \frac{1}{\ln 2} (\frac{1}{x^2}- \frac{1}{(1-x)^2}),
  \]
  we have $\abs{H^{(3)}(\tau + a)} \leq 4 \frac{1}{\ln 2} \max(\frac{1}{\tau^2}, \frac{1}{(1-\tau)^2})$, and similarly for $b$.  Hence
  \begin{align}
    \frac{1}{2\eps} \left(H(\tau + \eps) - H(\tau - \eps)\right)= H'(\tau) + \Delta\label{eq:Delta}
  \end{align}
  for
  $\abs{\Delta} \leq \frac{\eps^2}{6} \frac{8}{\ln 2}
  \max(\frac{1}{\tau^2}, \frac{1}{(1-\tau)^2})$.  Returning
  to~\eqref{eq:objective},
  \[
    H^\prime(\tau + (2q - 1)\eps) = H^\prime(\tau) + (2q - 1) \eps H^{(2)}(\tau + c)
  \]
  for some $\abs{c} \leq \abs{2q-1}\eps$.  Combining with~\eqref{eq:Delta},
  \[
    \abs{(2q-1) \eps H^{(2)}(\tau + c) } = \abs{\Delta} \leq
    \eps^2\frac{4}{3\ln 2} \max(\frac{1}{\tau^2},
    \frac{1}{(1-\tau)^2}).
  \]
  Since $H^{(2)}(\tau + c) = \frac{1}{\ln 2 (\tau + c) (1 - \tau - c)} \geq \frac{2}{3} \frac{1}{\ln 2 \cdot \tau (1-\tau)}$,
  \[
    \abs{2q-1} \leq 2 \eps \tau(1-\tau)\max(\frac{1}{\tau^2},
    \frac{1}{(1-\tau)^2}) = 2 \eps \max( \frac{1-\tau}{\tau}, \frac{\tau}{1-\tau}) \leq \frac{2\eps}{\tau(1-\tau)}.
  \]
\end{proof}

\begin{lemma}
\label{lem:capacityepsprime}
	\[
		C_{\tau,(1-o(1))\eps} \geq (1 - o(1))C_{\tau,\eps}
	\]
\end{lemma}

\begin{proof}
	Let $C^\prime(\alpha) := \min((\tau + \alpha\epsilon)\lg(\frac{\tau+\epsilon}{\tau + (2q - 1)\eps}) + (1 - \tau - \alpha\epsilon)\lg(\frac{1 - \tau - \epsilon}{1 - \tau - (2q - 1)\eps}), 
	(\tau - \alpha\epsilon)\lg(\frac{\tau-\epsilon}{\tau + (2q - 1)\eps}) + (1 - \tau + \alpha\epsilon)\lg(\frac{1 - \tau + \epsilon}{1 - \tau - (2q - 1)\eps}))$ for $0 < \alpha < 1$.

	Note that $C^\prime(\alpha) \leq C_{\tau,\alpha\eps}$ as otherwise we can use the analysis of Lemma \ref{lem:keylemma} to show that, we can solve NBS when $\eps^\prime = \alpha\epsilon$ with $\delta =1/n^\frac{1}{\sqrt{\ln{n}}}$, in  $(1 + o(1))\frac{\lg(n)}{(C^\prime(\alpha) + \capacity)/2}$ samples, which contradicts our lower bound.
	 
	Therefore
	\begin{align*}
		C_{\tau,\eps} - C^\prime(\alpha) &\leq C^\prime(1) - C^\prime(\alpha)\\
		&\leq \begin{aligned}\max( (1-\alpha)\epsilon (\lg(\frac{\tau + \eps}{\tau + (2q - 1)\eps}) - \lg(\frac{1 - \tau - \eps}{1 - \tau - (2q - 1)\eps}) ),\\(1-\alpha)\epsilon (\lg(\frac{\tau + (2q - 1)\eps}{\tau - \eps}) - \lg(\frac{1 - \tau - (2q - 1)\eps}{1 - \tau + \eps}) ) )\end{aligned}\\
		&\leq \begin{aligned}\max( (1-\alpha)\epsilon (\lg(\frac{\tau + \eps}{\tau - \eps}) + \lg(\frac{1 - \tau + \eps}{1 - \tau - \eps}) ),\\(1-\alpha)\epsilon (\lg(\frac{\tau + \eps}{\tau - \eps}) + \lg(\frac{1 - \tau + \eps}{1 - \tau - \eps}) ) )\end{aligned}\\
		&\leq \begin{aligned} (1-\alpha)\epsilon (\lg(\frac{\tau + \eps}{\tau - \eps}) + \lg(\frac{1 - \tau + \eps}{1 - \tau - \eps})\end{aligned}\\
		&\leq (1-\alpha)\eps (O(\frac{\eps}{\tau}) + O(\frac{\eps}{1 - \tau}) )\\
		&\leq (1-\alpha)O(\frac{\eps^2}{\tau(1-\tau)})\\
		&\leq (1-\alpha)O(C_{\tau,\eps})\\
	\end{align*}

	So when $\alpha = 1 - o(1)$

	\begin{align*}
		C_{\tau,\eps} - C_{\tau,\alpha\eps} &\leq (1-\alpha)O(C_{\tau,\eps})\\ 
		&\leq o(1)O(C_{\tau,\eps})\\ 
		C_{\tau,\eps} - o(1)O(C_{\tau,\eps}) &\leq  C_{\tau,\alpha\eps}\\ 
		(1-  o(1))C_{\tau,\eps} &\leq C_{\tau,\alpha\eps}\\ 
	\end{align*}

\end{proof}

\begin{lemma}
	\label{lem:logdbounds}
	For $\eps < \frac{1}{2}\min(\tau, 1-\tau)$,
	\begin{itemize}
		\item $\lg d_{0,0} \geq -3 \frac{\eps}{1 - \tau}$
		\item $\lg \frac{1}{d_{0,1}} \geq -3 \frac{\eps}{1 - \tau}$
		\item $\lg \frac{1}{d_{1,0}} \geq -3 \frac{\eps}{\tau}$
		\item $\lg d_{1,1} \geq - 3\frac{\eps}{\tau}$
	\end{itemize}
\end{lemma}

\begin{proof}
	When $x \in [0,1/4], \lg \frac{1 - x}{1 + x} \geq 4x \cdot\lg \frac{1 - .25}{1 + .25} \geq -3x$.
	\begin{itemize}
	\item $\lg d_{0,0} = \lg \expdoo \geq \lg \frac{1 - \tau - \eps}{1 - \tau + \eps} = \lg \frac{1 - \frac{\eps}{1 - \tau}}{1 + \frac{\eps}{1 - \tau}} \geq -3 \frac{\eps}{1 - \tau}$
	\item $\lg \frac{1}{d_{0,1}} = \lg \frac{1 - \tau - (2q - 1)\eps}{1 - \tau + \eps} \geq \lg \frac{1 - \tau - \eps}{1 - \tau + \eps} = \lg \frac{1 - \frac{\eps}{1 - \tau}}{1 + \frac{\eps}{1 - \tau}} \geq -3 \frac{\eps}{1 - \tau}$
	\item $\lg \frac{1}{d_{1,0}} = \lg \frac{\tau + (2q - 1)\eps}{\tau + \eps} \geq \lg \frac{\tau - \eps}{\tau + \eps} = \lg \frac{1 - \frac{\eps}{\tau}}{1 + \frac{\eps}{\tau}} \geq -3 \frac{\eps}{\tau}$
	\item $\lg d_{1,1} = \lg \expdll \geq \lg \frac{\tau - \eps}{\tau + \eps} = \lg \frac{1 - \frac{\eps}{\tau}}{1 + \frac{\eps}{\tau}} \geq -3 \frac{\eps}{\tau}$
	\end{itemize}
\end{proof}

\end{document}